\documentclass[11pt]{article}

\usepackage[margin=1in]{geometry}
\usepackage{amsmath, amsthm, hyperref, CJK}

\newtheorem{corollary}{Corollary}
\newtheorem{lemma}{Lemma}
\newtheorem{theorem}{Theorem}

\theoremstyle{definition}
\newtheorem{definition}{Definition}

\theoremstyle{remark}
\newtheorem*{remark}{Remark}

\DeclareMathOperator{\diag}{diag}
\DeclareMathOperator{\poly}{poly}
\DeclareMathOperator{\Span}{span}
\DeclareMathOperator{\tr}{tr}
\DeclareMathOperator{\trunc}{trunc}

\begin{document}

\begin{CJK*}{UTF8}{}

\title{Computing energy density in one dimension}

\CJKfamily{gbsn}

\author{Yichen Huang (黄溢辰)\\
Department of Physics, University of California, Berkeley, Berkeley, California 94720, USA\\
yichenhuang@berkeley.edu}

\maketitle

\end{CJK*}

\begin{abstract}

We study the problem of computing energy density in one-dimensional quantum systems. We show that the ground-state energy per site or per bond can be computed in time (i) independent of the system size and subexponential in the desired precision if the ground state satisfies area laws for the Renyi entanglement entropy (this is the first rigorous formulation of the folklore that area laws imply efficient matrix-product-state algorithms); (ii) independent of the system size and polynomial in the desired precision if the system is gapped. As a by-product, we prove that in the presence of area laws (or even an energy gap) the ground state can be approximated by a positive semidefinite matrix product operator of bond dimension independent of the system size and subpolynomial in the desired precision of local properties.

\end{abstract}

\section{Introduction} \label{intro}

The quantum PCP conjecture \cite{AAV13} is an assertion about the computational complexity of the ground-state energy per bond on general interaction graphs \cite{BH13}. It is one of the most prominent open problems for theoretical computer scientists in the emerging field of quantum Hamiltonian complexity \cite{Osb12}. Furthermore, computing energy density on regular lattices is a fundamental problem in condensed matter physics. Here we present a comprehensive study of this problem in one-dimensional (1D) quantum systems.

Consider a Hamiltonian $H=\sum_{i=1}^{n-1}H_i$ on a chain of $n$ spins (qudits), where $H_i$ acts on the spins $i$ and $i+1$ (nearest-neighbor interaction). Without loss of generality, assume $\|H_i\|\le1$ and $\lambda(H_i)=0$, where $\lambda(\cdots)$ denotes the ground-state energy of a Hamiltonian. We would like to compute the energy density $\lambda(H)/n$ to precision $\delta$, where $\delta=\Theta(1)$ is an (arbitrarily small) constant. We begin with a general result.

\begin{lemma} \label{lem0}
$\lambda(H)/n\pm\delta$ can be computed in time $2^{O(1/\delta)}$ with probability at least $0.999$.
\end{lemma}

\begin{proof}
We divide the system into $n\delta/2$ blocks, each of which consists of $2/\delta$ spins (we assume both $n\delta/2$ and $2/\delta$ are integers for convenience). In particular, define $H'=\sum_{i=1}^{n\delta/2}H_i'$ with $H_i'=\sum_{j=1}^{2/\delta-1}H_{2i/\delta+j}$ acting on the spins $2i/\delta+1,2i/\delta+2,\ldots,2(i+1)/\delta$. Let $i_1,i_2,\ldots,i_m$ with $m=O(1/\delta^2)$ be independent uniform random variables on $\{1,2,\ldots,n\delta/2\}$. The Chernoff bound implies
\begin{eqnarray}
&&\|H-H'\|=\left\|\sum_{i=1}^{n\delta/2}H_{2i/\delta}\right\|\le n\delta/2\Rightarrow0\le\lambda(H)/n-\lambda(H')/n\le\delta/2\nonumber\\
&&\Rightarrow\Pr\left(\left|\lambda(H)/n-\frac{\delta}{2m}\sum_{k=1}^m\lambda(H_{i_k}')\right|\ge\delta\right)\le\Pr\left(\left|\lambda(H')/n-\frac{\delta}{2m}\sum_{k=1}^m\lambda(H_{i_k}')\right|\ge\delta/2\right)\nonumber\\
&&\le\exp(-\Omega(m\delta^2))\le0.001.
\end{eqnarray}
As each $\lambda(H_{i_k}')$ can be computed by exact diagonalization, and the running time of the algorithm is $m\exp(O(1/\delta))=\exp(O(1/\delta))$.
\end{proof}

\begin{remark}
This algorithm is an efficient polynomial-time randomized approximation scheme (EPRAS) but not a fully polynomial-time randomized approximation scheme (FPRAS). Here, a randomized $f(1/\delta)\poly(n)$-time algorithm is an EPRAS ($f$ can be any function), and a randomized $\poly(n/\delta)$-time algorithm is an FPRAS. As computing $\lambda(H)$ to precision $\delta=1/\poly(n)$ is QMA-complete \cite{AGIK07, AGIK09}, an FPRAS in general 1D quantum systems is believed to be impossible.
\end{remark}

The algorithm of Lemma \ref{lem0} is inefficient both in theory and in practice because its running time grows exponentially with $1/\delta$. In this paper, we improve the scaling in $1/\delta$ in the presence of area laws for entanglement or even an energy gap.

\section{Computing energy density in the presence of area laws} \label{area}

A folklore in condensed matter physics is that in 1D a (ground) state can be efficiently represented by a matrix product state (MPS) if it satisfies an area law for entanglement; then, it is likely that the MPS representation of the ground state can be found efficiently using the (heuristic) DMRG algorithm \cite{Whi92, Whi93}. In this section, we give the first rigorous formulation of this folklore. In particular, we show that the ground-state energy density $\lambda(H)/n$ can be computed to precision $\delta$ in time independent of $n$ and subexponential in $1/\delta$ if the ground state satisfies area laws for the Renyi entanglement entropy (a function $f(x)$ is subexponential if $f(x)=o(\exp(x^c))$ for any constant $c>0$). As a by-product, we prove that in the presence of area laws the ground state can be approximated by a positive semidefinite matrix product operator (MPO) of bond dimension independent of $n$ and subpolynomial in the precision of local properties.

We use the Renyi entanglement entropy because an area law for (or even logarithmic divergence of) the Renyi entanglement entropy $R_\alpha(0<\alpha<1)$ implies efficient MPS representations \cite{VC06}. In contrast, an area law for the von Neumann entanglement entropy does not necessarily imply efficient MPS representations (see \cite{SWVC08} for a counterexample), although the von Neumann entanglement entropy is the most popular measure of entanglement (for pure states) in quantum information and condensed matter theory.

\subsection{Preliminaries} \label{pre}

MPS is a data structure underlying DMRG-like algorithms.

\begin{definition} [Matrix product state (MPS) \cite{PVWC07, FNW92}]
Let $d=\Theta(1)$ be the local dimension of each spin, and $\{|j_k\rangle\}_{j_k=1}^d$ be the computational basis of the Hilbert space of the spin $k$. Suppose $\{D_k\}_{k=1}^{n+1}$ with $D_1=D_{n+1}=1$ is a sequence of positive integers. An MPS with open boundary conditions takes the form
\begin{equation} \label{MPS}
|\psi\rangle=\sum_{j_1,j_2,\ldots,j_n=1}^dA_{j_1}^{[1]}A_{j_2}^{[2]}\cdots A_{j_n}^{[n]}|j_1j_2\cdots j_n\rangle,
\end{equation}
where $A_{j_k}^{[k]}$ is a matrix of size $D_k\times D_{k+1}$. Define $D=\max\{D_k\}_{k=1}^{n+1}$ as the bond dimension of the MPS $|\psi\rangle$.
\end{definition}

\begin{remark}
Any state can be expressed exactly as an MPS of bond dimension $D=\exp(O(n))$, and an MPS representation is efficient if $D=\poly(n)$.
\end{remark}

Any MPS can be transformed into the so-called canonical form \cite{PVWC07} with the same bond dimension such that
\begin{equation} \label{canonical}
\sum_{j=1}^dA_j^{[k]}A_j^{[k]\dag}=I,~\sum_{j=1}^dA_j^{[k]\dag}\Lambda^{[k]}A_j^{[k]}=\Lambda^{[k+1]},
\end{equation}
where $\Lambda^{[k]}=\diag\{(\lambda_1^{[k]})^2,(\lambda_2^{[k]})^2,\ldots\}$ with $\lambda_1^{[k]}\ge\lambda_2^{[k]}\ge\cdots>0$ the Schmidt coefficients across the cut $k-1|k$ in nonascending order. In the remainder of this section, we will use the canonical condition (\ref{canonical}) extensively.

MPO is the operator analog of MPS.

\begin{definition} [Matrix product operator (MPO)]
Let $\{\sigma_{j_k}\}_{j_k=1}^{d^2}$ be a basis of the space of operators on the spin $k$. Suppose $\{D_k\}_{k=1}^{n+1}$ with $D_1=D_{n+1}=1$ is a sequence of positive integers. An MPO with open boundary conditions takes the form
\begin{equation}
K=\sum_{j_1,j_2,\ldots,j_n=1}^{d^2}\left(A_{j_1}^{[1]}A_{j_2}^{[2]}\cdots A_{j_n}^{[n]}\right)\sigma_{j_1}\otimes\sigma_{j_2}\otimes\cdots\otimes\sigma_{j_n},
\end{equation}
where $A_{j_k}^{[k]}$ is a matrix of size $D_k\times D_{k+1}$. Similarly, define $D=\max\{D_k\}_{k=1}^{n+1}$ as the bond dimension of the MPO $K$.
\end{definition}

\begin{remark}
Any operator can be expressed exactly as an MPO of bond dimension $D=\exp(O(n))$.
\end{remark}

\begin{definition} [Entanglement entropy]
The Renyi entanglement entropy $R_\alpha(0<\alpha<1)$ of a bipartite (pure) quantum state $\rho_{AB}$ is defined as
\begin{equation}
R_\alpha(\rho_A)=(1-\alpha)^{-1}\log\tr\rho_A^\alpha,
\end{equation}
where $\rho_A=\tr_B\rho_{AB}$ is the reduced density matrix. The von Neumann entanglement entropy is defined as
\begin{equation}
S(\rho_A)=-\mathrm{tr}(\rho_A\log\rho_A)=\lim_{\alpha\rightarrow1^-}R_\alpha(\rho_A).
\end{equation}
\end{definition}

\begin{remark}
As $R_\alpha$ is a monotonically decreasing function of $\alpha$, an area law for $R_\alpha$ implies that for $R_\beta$ if $\alpha\le\beta$.
\end{remark}

Fixing a cut, the truncation error can be upper bounded as follows.

\begin{lemma}[\cite{VC06}] \label{lem1}
Let $\lambda_1\ge\lambda_2\ge\cdots>0$ with $\sum_j\lambda_j^2=1$ be the Schmidt coefficients across the cut in nonascending order. Then,
\begin{equation}
\varepsilon_D:=\sum_{j\ge D+1}\lambda_j^2\le\exp((1-\alpha)(R_\alpha-\log D)/\alpha).
\end{equation}
\end{lemma}

Minimizing the energy over MPS of constant bond dimension can be done efficiently by dynamic programming (or the transfer matrix method in statistical mechanics).

\begin{lemma} [\cite{SC10, AAI10}] \label{lem2}
Let $\mathcal{MPS}_D$ be the set of all normalized MPS of bond dimension $D$, and $H$ be a local Hamiltonian on a chain of $n$ spins. Then $\min_{|\psi\rangle\in\mathcal{MPS}_D}\langle\psi|H|\psi\rangle$ can be computed to precision $1/\poly(n)$ in time $n^{O(D^2)}$.
\end{lemma}

\subsection{Approximation of local properties}

As errors may accumulate while truncating each bond, it is unlikely that the ground-state wave function can be approximated up to a small constant error (measured by fidelity) by an MPS of bond dimension independent of the system size, even if there is an energy gap. In this subsection, we prove that in the presence of area laws the ground state can be approximated by a positive semidefinite MPO of bond dimension independent of the system size and subpolynomial in the precision of local properties (a function $f(x)$ is subpolynomial if $f(x)=o(x^c)$ for any constant $c>0$, i.e., it grows slower than any power of $x$).

\begin{lemma} \label{lem3}
There exists a positive semidefinite MPO $\rho_D$ of bond dimension $D^2$ such that
\begin{equation} \label{t1}
\left|\langle\psi|\hat O|\psi\rangle-\tr\left(\rho_D\hat O\right)\right|\le O(\sqrt{\varepsilon_D}),~\varepsilon_D:=\max_k\sum_{j\ge D+1}\left(\lambda^{[k]}_j\right)^2
\end{equation}
for any local operator $\hat O$ with bounded norm, where $\{\lambda^{[k]}_j\}_{j\ge1}$ are the Schmidt coefficients of $|\psi\rangle$ across the cut $k-1|k$ in nonascending order.
\end{lemma}

\begin{proof}
We express $|\psi\rangle$ exactly as an MPS (\ref{MPS}) of bond dimension $\exp(O(n))$ in the canonical form (\ref{canonical}), and define an MPO $\rho_D$ of bond dimension $D^2$ as
\begin{equation} \label{mpo}
\rho_D=\sum_{i_1,j_1,i_2,j_2,\ldots,i_n,j_n=1}^dB_{i_1,j_1}^{[1]}B_{i_2,j_2}^{[2]}\cdots B_{i_n,j_n}^{[n]}|i_1i_2\cdots i_n\rangle\langle j_1j_2\cdots j_n|.
\end{equation}
In the bulk, $B_{i_k,j_k}^{[k]}$ is a matrix of size $D^2\times D^2$; at the boundaries, $B_{i_1,j_1}^{[1]}$ and $B_{i_n,j_n}^{[n]}$ are row and column vectors of length $D^2$, respectively. Without loss of generality, we will only work in the bulk in the remainder of this proof. The rows of $B_{i,j}^{[k]}$ are labeled by two indices $\alpha,\alpha'$ for $\alpha,\alpha'=1,2,\ldots,D$, and the columns of $B_{i,j}^{[k]}$ are labeled by $\beta,\beta'$ for $\beta,\beta'=1,2,\ldots,D$. Let $\Lambda^{[k]}_j:=(\lambda^{[k]}_j)^2$ and $S:=\{1,D+1,D+2,\ldots\}$. Define $x^{[k]}_s=\sqrt{\Lambda^{[k]}_s/\sum_{t\in S}\Lambda^{[k]}_t}$ if $s\in S$ and $x^{[k]}_s=1$ otherwise. The matrix elements of $B^{[k]}_{i,j}$ are given by
\begin{eqnarray} \label{def}
&&B^{[k]}_{i,j}(\alpha=\alpha'=1;\beta=\beta'=1)=\sum_{s,t\in S}A^{[k]}_i(s,t)A^{[k]*}_j(s,t)\left(x^{[k]}_s\right)^2,\nonumber\\
&&B^{[k]}_{i,j}(\alpha=\alpha'=1;\beta+\beta'\ge3)=\sum_{s\in S}A^{[k]}_i(s,\beta)A^{[k]*}_j(s,\beta')\left(x^{[k]}_s\right)^2,\nonumber\\
&&B^{[k]}_{i,j}(\alpha+\alpha'\ge3;\beta=\beta'=1)=\sum_{t\in S}A^{[k]}_i(\alpha,t)A^{[k]*}_j(\alpha',t)x^{[k]}_\alpha x^{[k]}_{\alpha'},\nonumber\\
&&B^{[k]}_{i,j}(\alpha+\alpha'\ge3;\beta+\beta'\ge3)=A^{[k]}_i(\alpha,\beta)A^{[k]*}_j(\alpha',\beta')x^{[k]}_\alpha x^{[k]}_{\alpha'},
\end{eqnarray}
where (e.g.) $A^{[k]}_i(s,t)$ is the element in the $s$th row and $t$th column of $A^{[k]}_i$. Define $\tilde\Lambda^{[k]}_s=\sum_{t\in S}\Lambda^{[k]}_t$ if $s\in S$ and $\tilde\Lambda^{[k]}_s=\Lambda^{[k]}_s$ otherwise. The matrices $B^{[k]}_{i,j}$'s satisfy a pair of conditions similar to (\ref{canonical}):
\begin{equation}
\sum_{\beta,\beta',i}B^{[k]}_{i,i}(\alpha,\alpha';\beta,\beta')\delta_{\beta\beta'}=\delta_{\alpha\alpha'},~\sum_{\alpha,\alpha',i}B^{[k]}_{i,i}(\alpha,\alpha';\beta,\beta')\delta_{\alpha\alpha'}\tilde\Lambda^{[k]}_{\alpha}=\delta_{\beta\beta'}\tilde\Lambda^{[k+1]}_{\beta}.
\end{equation}
Indeed,
\begin{eqnarray}
&&\sum_{\beta,\beta',i}B^{[k]}_{i,i}(\alpha=\alpha'=1;\beta,\beta')\delta_{\beta\beta'}=\sum_{s\in S,t,i}A^{[k]}_i(s,t)A^{[k]*}_i(s,t)\left(x^{[k]}_s\right)^2=\sum_{s\in S}\left(x^{[k]}_s\right)^2=1,\\
&&\sum_{\beta,\beta',i}B^{[k]}_{i,i}(\alpha+\alpha'\ge3;\beta,\beta')\delta_{\beta\beta'}=\sum_{t,i}A^{[k]}_i(\alpha,t)A^{[k]*}_i(\alpha',t)x^{[k]}_\alpha x^{[k]}_{\alpha'}=\delta_{\alpha\alpha'}x^{[k]}_\alpha x^{[k]}_{\alpha'}=\delta_{\alpha\alpha'},\\
&&\sum_{\alpha,\alpha',i}B^{[k]}_{i,i}(\alpha,\alpha';\beta=\beta'=1)\delta_{\alpha\alpha'}\tilde\Lambda^{[k]}_{\alpha}=\sum_{t\in S,s,i}A^{[k]}_i(s,t)A^{[k]*}_i(s,t)\left(x^{[k]}_s\right)^2\tilde\Lambda^{[k]}_s\nonumber\\
&&=\sum_{t\in S,s,i}A^{[k]}_i(s,t)A^{[k]*}_i(s,t)\Lambda^{[k]}_s=\tilde\Lambda^{[k+1]}_1,\\
&&\sum_{\alpha,\alpha',i}B^{[k]}_{i,i}(\alpha,\alpha';\beta+\beta'\ge3)\delta_{\alpha\alpha'}\tilde\Lambda^{[k]}_{\alpha}=\sum_{s,i}A^{[k]}_i(s,\beta)A^{[k]*}_i(s,\beta')\left(x^{[k]}_s\right)^2\tilde\Lambda^{[k]}_s\nonumber\\
&&=\sum_{s,i}A^{[k]}_i(s,\beta)A^{[k]*}_i(s,\beta')\Lambda^{[k]}_s=\delta_{\beta\beta'}\tilde\Lambda^{[k+1]}_\beta.
\end{eqnarray}
Assume without loss of generality that $\hat O$ is a $2$-local operator on the spins $k-1$ and $k$. Let $\{|i^{[k]}\rangle\}_{i=1}^d$ be the computational basis of the Hilbert space of the spin $k$. The matrix elements of the reduced density matrix $\rho^{[k-1,k]}_D$ of $\rho_D$ on the support of $\hat O$ are
\begin{eqnarray}
&&\rho^{[k-1,k]}_D(j,j';i,i'):=\left\langle j^{[k-1]}j'^{[k]}\right|\rho_D^{[k-1,k]}\left|i^{[k-1]}i'^{[k]}\right\rangle\nonumber\\
&&=\sum_{s',t'\in S,s,t\ge1}A^{[k-1]}_i(s,t')A^{[k-1]*}_j(s,t')\Lambda^{[k-1]}_sA^{[k]}_{i'}(s',t)A^{[k]*}_{j'}(s',t)\left(x^{[k]}_{s'}\right)^2\nonumber\\
&&+\sum_{\alpha+\beta\ge3,s,t}A^{[k-1]}_i(s,\alpha)A^{[k-1]*}_j(s,\beta)\Lambda^{[k-1]}_sA^{[k]}_{i'}(\alpha,t)A^{[k]*}_{j'}(\beta,t)x^{[k]}_\alpha x^{[k]}_\beta\nonumber\\
&&=s_0^{j,j';i,i'}+s_0^{j,j';i,i'},
\end{eqnarray}
where
\begin{eqnarray}
&&s_0^{j,j';i,i'}=\sum_{s,t,\alpha,\beta}A^{[k-1]}_i(s,\alpha)A^{[k-1]*}_j(s,\beta)\Lambda^{[k-1]}_sA^{[k]}_{i'}(\alpha,t)A^{[k]*}_{j'}(\beta,t)x^{[k]}_\alpha x^{[k]}_\beta,\\
&&s_1^{j,j';i,i'}=\sum_{s',t'\in S,s'+t'\ge3,s,t}A^{[k-1]}_i(s,t')A^{[k-1]*}_j(s,t')\Lambda^{[k-1]}_sA^{[k]}_{i'}(s',t)A^{[k]*}_{j'}(s',t)\left(x^{[k]}_{s'}\right)^2.
\end{eqnarray}
Let
\begin{equation}
L^{i,i'}_{s,t,s',t'}:=A^{[k-1]}_i(s,t')\sqrt{\Lambda^{[k-1]}_s}A^{[k]}_{i'}(s',t)x^{[k]}_{s'}
\end{equation}
such that
\begin{eqnarray}
&&s_1^{j,j';i,i'}=\sum_{s',t'\in S,s'+t'\ge3,s,t}L^{i,i'}_{s,t,s',t'}L^{j,j'*}_{s,t,s',t'}\Rightarrow\left|s_1^{j,j';i,i'}\right|\le\frac{1}{2}\sum_{s',t'\in S,s'+t'\ge3,s,t}\left|L^{i,i'}_{s,t,s',t'}\right|^2+\left|L^{j,j'}_{s,t,s',t'}\right|^2\nonumber\\
&&\Rightarrow\sum_{i,i',j,j'}\left|s_1^{j,j';i,i'}\right|\le\sum_{s',t'\in S,s'+t'\ge3,s,t,i,i',j,j'}L^{i,i'}_{s,t,s',t'}L^{i,i'*}_{s,t,s',t'}\nonumber\\
&&=d^2\sum_{s',t'\in S,s'+t'\ge3,s,t,i,i'}A^{[k-1]}_i(s,t')\Lambda^{[k-1]}_sA^{[k-1]*}_i(s,t')A^{[k]}_{i'}(s',t)A^{[k]*}_{i'}(s',t)\left(x^{[k]}_{s'}\right)^2\nonumber\\
&&=d^2\sum_{s',t'\in S,s'+t'\ge3}\Lambda^{[k]}_{t'}\left(x^{[k]}_{s'}\right)^2=\frac{d^2\sum_{s',t'\in S,s'+t'\ge3}\Lambda^{[k]}_{t'}\Lambda^{[k]}_{s'}}{\sum_{r\in S}\Lambda^{[k]}_r}\le2d^2\sum_{s\ge D+1}\Lambda^{[k]}_s\le2d^2\varepsilon_D.
\end{eqnarray}
Let $P=\diag\{1,1,\ldots,1,0,0,\ldots\}$ with $\tr P=D$ and $P_1=\diag\{1,0,0,\ldots\}\le P$. The matrix elements of the reduced density matrix $\rho^{[k-1,k]}$ of $|\psi\rangle\langle\psi|$ on the support of $\hat O$ are
\begin{eqnarray}
&&\rho^{[k-1,k]}(j,j';i,i'):=\left\langle j^{[k-1]}j'^{[k]}\right|\rho^{[k-1,k]}\left|i^{[k-1]}i'^{[k]}\right\rangle=\tr\left(A_j^{[k-1]\dag}\Lambda^{[k-1]}A_i^{[k-1]}A_{i'}^{[k]}A_{j'}^{[k]\dag}\right)\nonumber\\
&&=r_0^{j,j';i,i'}+r_1^{j,j';i,i'}+r_2^{j,j';i,i'},
\end{eqnarray}
where
\begin{eqnarray}
&&r_0^{j,j';i,i'}=\tr\left(A_j^{[k-1]\dag}\Lambda^{[k-1]}A_i^{[k-1]}PA_{i'}^{[k]}A_{j'}^{[k]\dag}P\right),\nonumber\\
&&r_1^{j,j';i,i'}=\tr\left(A_j^{[k-1]\dag}\Lambda^{[k-1]}A_i^{[k-1]}PA_{i'}^{[k]}A_{j'}^{[k]\dag}(1-P)\right),\nonumber\\
&&r_2^{j,j';i,i'}=\tr\left(A_j^{[k-1]\dag}\Lambda^{[k-1]}A_i^{[k-1]}(1-P)A_{i'}^{[k]}A_{j'}^{[k]\dag}\right).
\end{eqnarray}
Using the Cauchy-Schwarz inequality,
\begin{eqnarray}
&&\sum_{i,i',j,j'}\left|r_1^{i,i';j,j'}\right|=\sum_{i,i',j,j'}\left|\tr\left(\sqrt{\Lambda^{[k-1]}}A_i^{[k-1]}PA_{i'}^{[k]}A_{j'}^{[k]\dag}(1-P)A_j^{[k-1]\dag}\sqrt{\Lambda^{[k-1]}}\right)\right|\nonumber\\
&&\le\sqrt{\sum_{i,i',j,j'}\tr\left(A_{i'}^{[k]\dag}PA_i^{[k-1]\dag}\Lambda^{[k-1]}A_i^{[k-1]}PA_{i'}^{[k]}\right)}\nonumber\\
&&\times\sqrt{\sum_{i,i',j,j'}\tr\left(A_{j'}^{[k]\dag}(1-P)A_j^{[k-1]\dag}\Lambda^{[k-1]}A_j^{[k-1]}(1-P)A_{j'}^{[k]}\right)}\nonumber\\
&&\le\sqrt{d^2\tr(P\Lambda^{[k]}P)d^2\tr((1-P)\Lambda^{[k]}(1-P))}=d^2\sqrt{\sum_{s\le D}\Lambda^{[k]}_s\sum_{s\ge D+1}\Lambda^{[k]}_s}\le d^2\sqrt{\sum_{s\ge D+1}\Lambda^{[k]}_s}\nonumber\\
&&\le d^2\sqrt{\varepsilon_D}.
\end{eqnarray}
Similarly,
\begin{equation}
\sum_{i,i',j,j'}\left|r_2^{j,j';i,i'}\right|\le\sqrt{d^2\tr\Lambda^{[k]}d^2\tr((1-P)\Lambda^{[k]}(1-P))}=d^2\sqrt{\sum_{s\ge D+1}\Lambda^{[k]}_s}\le d^2\sqrt{\varepsilon_D}.
\end{equation}
Using the Cauchy-Schwarz inequality,
\begin{eqnarray}
&&r_0^{j,j';i,i'}=\sum_{\alpha,\beta\le D,s,t}A^{[k-1]}_i(s,\alpha)A^{[k-1]*}_j(s,\beta)\Lambda^{[k-1]}_sA^{[k]}_{i'}(\alpha,t)A^{[k]*}_{j'}(\beta,t)\nonumber\\
&&\Rightarrow\sum_{i,i',j,j'}\left|r_0^{j,j';i,i'}-s_0^{j,j';i,i'}\right|\nonumber\\
&&\le\sum_{\alpha,\beta\le D,s,t,i,i',j,j'}\left|A^{[k-1]}_i(s,\alpha)A^{[k-1]*}_j(s,\beta)\Lambda^{[k-1]}_sA^{[k]}_{i'}(\alpha,t)A^{[k]*}_{j'}(\beta,t)\right|\left(1-x^{[k]}_\alpha x^{[k]}_\beta\right)\nonumber\\
&&\le\sum_{\alpha,\beta\le D,s,t,i,i',j,j'}\left|A^{[k-1]}_i(s,\alpha)A^{[k-1]*}_j(s,\beta)\Lambda^{[k-1]}_sA^{[k]}_{i'}(\alpha,t)A^{[k]*}_{j'}(\beta,t)\right|\left(1-x^{[k]}_\alpha+1-x^{[k]}_\beta\right)\nonumber\\
&&=\left(1-x^{[k]}_1\right)\sum_{i,i',j,j'}\left|A^{[k-1]\dag}_j\Lambda^{[k-1]}A^{[k-1]}_iPA^{[k]}_{i'}A^{[k]\dag}_{j'}P_1\right|+\left|A^{[k-1]\dag}_j\Lambda^{[k-1]}A^{[k-1]}_iP_1A^{[k]}_{i'}A^{[k]\dag}_{j'}P\right|\nonumber\\
&&\le\left(1-x^{[k]}_1\right)\left(d^2\sqrt{\Lambda^{[k]}_1\sum_{s\le D}\Lambda^{[k]}_s}+d^2\sqrt{\Lambda^{[k]}_1\sum_{s\le D}\Lambda^{[k]}_s}\right)\le2d^2\sqrt{\Lambda^{[k]}_1}\left(1-x^{[k]}_1\right)\nonumber\\
&&=2d^2\lambda^{[k]}_1\left(1-\lambda^{[k]}_1\Big/\sqrt{\left(\lambda^{[k]}_1\right)^2+\varepsilon_D}\right)=2d^2\varepsilon_D\Big/\left(\lambda^{[k]}_1+\varepsilon_D/\lambda^{[k]}_1+\sqrt{\left(\lambda^{[k]}_1\right)^2+\varepsilon_D}\right)\nonumber\\
&&\le2d^2\varepsilon_D\Big/\left(2\lambda^{[k]}_1+\varepsilon_D/\lambda^{[k]}_1\right)\le d^2\sqrt{2\varepsilon_D}.
\end{eqnarray}
Finally, (\ref{t1}) follows from
\begin{eqnarray}
&&\sum_{i,i'j,j'}\left|\rho^{[k-1,k]}(j,j';i,i')-\rho^{[k-1,k]}_D(j,j';i,i')\right|\nonumber\\
&&\le\sum_{i,i',j,j'}\left|r_1^{j,j';i,i'}\right|+\left|r_2^{j,j';i,i'}\right|+\left|r_0^{j,j';i,i'}-s_0^{j,j';i,i'}\right|+\left|s_1^{j,j';i,i'}\right|\nonumber\\
&&\le d^2\sqrt{\varepsilon_D}+d^2\sqrt{\varepsilon_D}+d^2\sqrt{2\varepsilon_D}+2d^2\varepsilon_D=O(\sqrt{\varepsilon_D})\nonumber\\
&&\Rightarrow\left\|\rho^{[k-1,k]}-\rho^{[k-1,k]}_D\right\|=O(\sqrt{\varepsilon_D}).
\end{eqnarray}
as $d=\Theta(1)$ is an absolute constant.
\end{proof}

\begin{theorem} \label{thm1}
Suppose $|\psi\rangle$ satisfies an area law for the Renyi entanglement entropy $R_\alpha$ (across any cut). Then there exists an MPO $\rho$ of bond dimension $\exp(2R_\alpha)O(1/\delta)^{4\alpha/(1-\alpha)}$ such that
\begin{equation}
\left|\langle\psi|\hat O|\psi\rangle-\tr\left(\rho\hat O\right)\right|\le\delta
\end{equation}
for any local operator $\hat O$ with bounded norm.
\end{theorem}

\begin{proof}
This is an immediate consequence of Lemmas \ref{lem1}, \ref{lem3}. In particular,
\begin{equation}
\delta=O(\sqrt{\epsilon_D})\le O(\exp((1-\alpha)(R_\alpha-\log D)/(2\alpha)))\Rightarrow D^2=\exp(2R_\alpha)O(1/\delta)^{4\alpha/(1-\alpha)}.
\end{equation}
\end{proof}

\begin{remark}
If $|\psi\rangle$ satisfies area laws for the Renyi entanglement entropy $R_{\forall\alpha}$, then the bond dimension is independent of $n$ and subpolynomial in $1/\delta$.
\end{remark}

\subsection{Algorithm and analysis}

\begin{theorem} \label{thm2}
Suppose the ground state $|\psi\rangle$ of a local Hamiltonian $H$ on a chain of $n$ spins satisfies an area law for the Renyi entanglement entropy $R_\alpha$ (across any cut). Then there exists an MPS $|\phi\rangle$ of bond dimension $D=\exp(R_\alpha)O(1/\delta)^{2\alpha/(1-\alpha)}$ such that
\begin{equation} \label{t2}
\langle\phi|H|\phi\rangle\le\lambda(H)+n\delta.
\end{equation}
\end{theorem}

\begin{proof}
Recall that $\rho_D$ (\ref{mpo}) is an MPO of bond dimension $D^2$ such that for any local operator $\hat O$ with bounded norm,
\begin{equation} \label{ener}
\tr\left(\rho_D\hat O\right)\le\langle\psi|\hat O|\psi\rangle+\delta\Rightarrow\tr(\rho_DH)\le\langle\psi|H|\psi\rangle+n\delta.
\end{equation}
The matrices defined in (\ref{def}) can be expressed as a sum of many terms: $B^{[k]}_{i,j}=\sum_{p,q\ge0}C^{[k]}_{i,j,p,q}$, where $C^{[k]}_{i,j,p,q}$ is a matrix of size $D^2\times D^2$ (except at the boundaries). In particular, the matrix elements of $C^{[k]}_{i,j,p,q}$ are given by
\begin{eqnarray}
&&C^{[k]}_{i,j,p=0,q=0}(\alpha,\alpha';\beta,\beta')=A^{[k]}_i(\alpha,\beta)A^{[k]*}_j(\alpha',\beta')x^{[k]}_\alpha x^{[k]}_{\alpha'},\\
&&C^{[k]}_{i,j,p=0,q\ge1}(\alpha,\alpha';\beta=\beta'=1)=A^{[k]}_i(\alpha,D+q)A^{[k]*}_j(\alpha',D+q)x^{[k]}_\alpha x^{[k]}_{\alpha'},\\
&&C^{[k]}_{i,j,p\ge1,q=0}(\alpha=\alpha'=1;\beta,\beta')=A^{[k]}_i(D+p,\beta)A^{[k]*}_j(D+p,\beta')\left(x^{[k]}_{D+p}\right)^2,\\
&&C^{[k]}_{i,j,p\ge1,q\ge1}(\alpha=\alpha'=1;\beta=\beta'=1)=A^{[k]}_i(D+p,D+q)A^{[k]*}_j(D+p,D+q)\left(x^{[k]}_{D+p}\right)^2,\\
&&C^{[k]}_{i,j,p=0,q\ge1}(\alpha,\alpha';\beta+\beta'\ge3)=C^{[k]}_{i,j,p\ge1,q=0}(\alpha+\alpha'\ge3;\beta,\beta')\nonumber\\
&&=C^{[k]}_{i,j,p\ge1,q\ge1}(\alpha,\alpha';\beta+\beta'\ge3)=C^{[k]}_{i,j,p\ge1,q\ge1}(\alpha+\alpha'\ge3;\beta,\beta')=0.
\end{eqnarray}
Thus, the MPO $\rho_D$ can be expressed as a sum of many MPO of bond dimension $D^2$:
\begin{equation}
\rho_D=\sum_{p_1,q_1,p_2,q_2,\ldots,p_n,q_n}\rho_D^{p_1,q_1,p_2,q_2,\ldots,p_n,q_n},
\end{equation}
where
\begin{equation}
\rho_D^{p_1,q_1,p_2,q_2,\ldots,p_n,q_n}=\sum_{i_1,j_1,i_2,j_2,\ldots,i_n,j_n=1}^d\left(C_{i_1,j_1,p_1,q_1}^{[1]}C_{i_2,j_2,p_2,q_2}^{[2]}\cdots C_{i_n,j_n,p_n,q_n}^{[n]}\right)|i_1i_2\cdots i_n\rangle\langle j_1j_2\cdots j_n|.
\end{equation}
We observe that each $\rho_D^{p_1,q_1,p_2,q_2,\ldots,p_n,q_n}$ is the density matrix of an unnormalized (pure) MPS of bond dimension $D$. Hence (\ref{ener}) implies that one of them (still denoted by $\rho_D^{p_1,q_1,p_2,q_2,\ldots,p_n,q_n}$) satisfies
\begin{equation}
\tr(\rho_D^{p_1,q_1,p_2,q_2,\ldots,p_n,q_n}H)/\tr\rho_D^{p_1,q_1,p_2,q_2,\ldots,p_n,q_n}\le\langle\psi|H|\psi\rangle+n\delta,
\end{equation}
and we complete the proof by normalizing $\rho_D^{p_1,q_1,p_2,q_2,\ldots,p_n,q_n}$.
\end{proof}

Let $\tilde O(x):=O(x\poly\log x)$ hide a polylogarithmic factor.

\begin{theorem} \label{thm3}
Suppose the ground state $|\psi\rangle$ of a local Hamiltonian $H$ on a chain of $n$ spins satisfies an area law for the Renyi entanglement entropy $R_\alpha$ (across any cut). Then $\lambda(H)/n\pm\delta$ can be computed in time $\exp(\exp(2R_\alpha)\tilde O(1/\delta)^{4\alpha/(1-\alpha)})$ with probability at least $0.999$.
\end{theorem}

\begin{proof}
We divide the system into $n\delta/2$ blocks, each of which consists of $2/\delta$ spins. In particular, define $H'=\sum_{i=1}^{n\delta/2}H_i'$ with $H_i'=\sum_{j=1}^{2/\delta-1}H_{2i/\delta+j}$ acting on the spins $2i/\delta+1,2i/\delta+2,\ldots,2(i+1)/\delta$. Theorem \ref{thm2} implies an MPS $|\phi\rangle$ of bond dimension $D=\exp(R_\alpha)O(1/\delta)^{2\alpha/(1-\alpha)}$ such that
\begin{equation}
\langle\phi|H|\phi\rangle\le\lambda(H)+n\delta/2.
\end{equation}
The Schmidt decomposition of $|\phi\rangle$ across the cut $2i/\delta|2i/\delta+1$ implies
\begin{equation} \label{sch}
|\phi\rangle=\sum_{j=1}^D\lambda_j|L_j\rangle|R_j\rangle\Rightarrow\langle\phi|H_i'|\phi\rangle=\sum_{j=1}^D\lambda_j^2\langle R_j|H_i'|R_j\rangle\Rightarrow\exists k_i,~\mathrm{s.~t.}~\langle R_{k_i}|H_i'|R_{k_i}\rangle\le\langle\phi|H_i'|\phi\rangle.
\end{equation}
As $|R_{k_i}\rangle$ is an MPS of bond dimension $D$, the Schmidt decomposition across the cut $2(i+1)/\delta|2(i+1)/\delta+1$ implies
\begin{eqnarray}
&&|R_{k_i}\rangle=\sum_{j=1}^D\lambda_j'|L_j'\rangle|R_j'\rangle\Rightarrow\langle R_{k_i}|H_i'|R_{k_i}\rangle=\sum_{j=1}^D\lambda_j^2\langle L_j'|H_i'|L_j'\rangle\nonumber\\
&&\Rightarrow\exists k_i',~\mathrm{s.~t.}~\left\langle L_{k_i'}'\right|H_i'\left|L_{k_i'}'\right\rangle\le\langle R_{k_i}|H_i'|R_{k_i}\rangle\le\langle\phi|H_i'|\phi\rangle\nonumber\\
&&\Rightarrow\lambda(H)-n\delta/2\le\lambda(H')=\sum_{i=1}^{n\delta/2}\lambda(H_i')\le\sum_{i=1}^{n\delta/2}\left\langle L_{k_i'}'\right|H_i'\left|L_{k_i'}'\right\rangle\le\sum_{i=1}^{n\delta/2}\langle\phi|H_i'|\phi\rangle\nonumber\\
&&=\langle\phi|H'|\phi\rangle\le\langle\phi|H|\phi\rangle\le\lambda(H)+n\delta/2.
\end{eqnarray}
As $|L_{k_i'}'\rangle$ is an MPS of bond dimension $D$, Lemma \ref{lem2} implies that a number between $\lambda(H_i')$ and $\langle L_{k_i'}'|H_i'|L_{k_i'}'\rangle$ can be computed in time $(1/\delta)^{O(D^2)}$. Finally, a randomized algorithm with running time $O(1/\delta^2)(1/\delta)^{O(D^2)}=\exp(\exp(2R_\alpha)\tilde O(1/\delta)^{4\alpha/(1-\alpha)})$ follows from the proof of Lemma \ref{lem0}.
\end{proof}

\begin{remark}
If $|\psi\rangle$ satisfies area laws for the Renyi entanglement entropy $R_{\forall\alpha}$, then the running time is independent of $n$ and subexponential in $1/\delta$.
\end{remark}

\section{Computing energy density in the presence of an energy gap}

\subsection{Introduction}

As an energy gap implies area laws for the Renyi entanglement entropy $R_{\forall\alpha}$ \cite{Hua14}, all results in Section \ref{area} apply. In particular, to approximate local properties the bond dimension of the MPO is subpolynomial in $1/\delta$ (Theorem \ref{thm1}), and to compute the ground-state energy density the running time of the algorithm is subexponential in $1/\delta$ (Theorem \ref{thm3}). Indeed, the scaling in $1/\delta$ can be obtained precisely.

Let $|\Psi_0\rangle$ and $\epsilon$ be the ground state and the energy gap of $H$, respectively. Since we are interested in the regime $1/\delta$ is large, we assume $1/\delta\gg1/\epsilon$. Some $\epsilon$-dependent subpolynomial (e.g., $2^{\tilde O(\epsilon^{-1/4}\log^{3/4}(1/\delta))}$) and constant (e.g., $2^{2^{\tilde O(1/\epsilon)}}$) factors will appear below. If not dominant (e.g., accompanied with $\poly(1/\delta)$), depending on the context they may be neglected or kept for simplicity or clarity, respectively.

Let $\tilde\Omega(x):=\Omega(x/\poly\log x)$ hide a polylogarithmic factor. As a by-product of the proof of the area law for entanglement, the truncation error can be upper bounded as follows.

\begin{lemma} [\cite{Hua14}] \label{lem4}
Let $\{\lambda_i\}$ be the Schmidt coefficients of $|\Psi_0\rangle$ across a cut. Then,
\begin{equation}
D=2^{\tilde O(1/\epsilon+\epsilon^{-1/4}\log^{3/4}(1/\varepsilon_D))},~\varepsilon_D:=\sum_{i\ge D+1}\lambda_i^2.
\end{equation}
\end{lemma}

Then, it is straightforward to obtain a pair of corollaries.

\begin{corollary}
There exists a positive semidefinite MPO $\rho$ of bond dimension $2^{\tilde O(1/\epsilon+\epsilon^{-1/4}\log^{3/4}(1/\delta))}$ such that
\begin{equation}
\left|\langle\psi|\hat O|\psi\rangle-\tr\left(\rho\hat O\right)\right|\le\delta
\end{equation}
for any local operator $\hat O$ with bounded norm.
\end{corollary}

\begin{corollary}
$\lambda(H)/n\pm\delta$ can be computed in time $2^{2^{\tilde O(\epsilon^{-1/4}\log^{3/4}(1/\delta))}}$ with probability at least $0.999$.
\end{corollary}

Our goal is to improve the scaling in $1/\delta$. The main result of this section is

\begin{theorem} [FPRAS for computing ground-state energy density of 1D gapped Hamiltonians]
$\lambda(H)/n\pm O(\delta)$ can be computed in time $(1/\delta)^{O(1)}$ with probability at least $0.999$.
\end{theorem}

\begin{proof}
Recall $H=\sum_{i=1}^{n-1}H_i$, where $0\le H_i\le1$ acts on the spins $i$ and $i+1$. Following the proof of Lemma \ref{lem0}, we divide $H$ the system into $n\delta$ blocks, each of which consists of $1/\delta$ spins. In particular, define $H'=\sum_{i=1}^{n\delta}H_i'$ with $H_i'=\sum_{j=1}^{1/\delta-1}H_{i/\delta+j}$ acting on the spins $i/\delta+1,i/\delta+2,\ldots,(i+1)/\delta$. Then, it suffices to prove Lemma \ref{lem5}.
\end{proof}

Let $\mathcal H$ be the Hilbert space $(\mathbf C^d)^{\otimes n}$ of the system ($n$ qudits) and $\mathcal H_{[i_l,i_r]}=(\mathbf C ^d)^{\otimes(i_r-i_l+1)}$ be the Hilbert space of the spins with indices in the interval $[i_l,i_r]$. Let $\epsilon_L=\lambda(\sum_{j=1}^{i_l-s-1}H_j)$ and $\epsilon_R=\lambda(\sum_{j=i_r+s+1}^{n-1}H_j)$. Define
\begin{equation} 
H_L^{\le t}:=\left(\sum_{j=1}^{i_l-s-1}H_j-\epsilon_L\right)P^{\le t}_L+t(1-P^{\le t}_L),
\end{equation}
where $P^{\le t}_L$ is the projection onto the subspace spanned by the eigenstates of $\sum_{j=1}^{i_l-s-1}H_j$ with energies at most $\epsilon_L+t$; similarly we define $H_R^{\le t}$. Let
\begin{equation} \label{truncated}
H^{\le t}:=H_L^{\le t}+H_{i_l-s}+H_{i_l-s+1}+\cdots+ H_{i_r+s}+H_R^{\le t}\le i_r-i_l+2s+2t+1.
\end{equation}
be a 1D Hamiltonian which is local on $\mathcal H_{[i_l-s+1,i_r+s]}$. Let $|\Psi_0'\rangle$ and $\epsilon'$ be the ground state and the energy gap of $H^{\le t}$, respectively.

\begin{lemma} [\cite{Hua14}]
For $t\ge O(\log\epsilon^{-1})$,\\
(a) $0\le\lambda(H)-\lambda(H^{\le t})-\epsilon_L-\epsilon_R\le2^{-\Omega(t)}$;\\
(b) $|\langle\Psi_0'|\Psi_0\rangle|\ge1-2^{-\Omega(t)}$;\\
(c) $\epsilon'\ge\epsilon/10$.
\end{lemma}

\begin{definition} [support set]
For any interval $[i_l,i_r]$, $S\subseteq\mathcal H_{[i_l,k]}$ with $k\le i_r$ is a $(k,s,b,\Delta_e~\mathrm{or}~\Delta_f)$-support set if there exists a state $|\psi\rangle\in\mathcal H$ (called a witness for $S$) such that\\
(i) the reduced density matrix of $|\psi\rangle$ on $\mathcal H_{[i_l,k]}$ is supported on $\Span S$;\\
(ii) $|S|\le s$;\\
(iv) all elements in $S$ are MPS of bond dimension at most $b$;\\
(iii) $|\langle\psi|\Psi_0'\rangle|\ge1-\Delta_f$ or $\langle\psi|H^{\le t}|\psi\rangle\le\lambda(H^{\le t})+\Delta_e\epsilon'$ (depending on the context either $\Delta_e$ or $\Delta_f$ is used as the precision parameter).
\end{definition}

\begin{lemma} \label{lem5}
A number $e_i$ between $\lambda(H_i')$ and $\lambda(H_i')+3$ can be computed in time $(1/\delta)^{O(1)}$.
\end{lemma}

\begin{proof} [Proof sketch]
Set $i_l=i/\delta+1$ and $i_r=(i+1)/\delta$. Our algorithm is a ``local'' variant of and thus very similar to the polynomial-time algorithm for computing the ground-state wave function in 1D gapped systems \cite{LVV13, Hua14, CF15}. In particular, it iteratively constructs a $(k,p_1p_3,p_2p_3,\Delta_e=c\epsilon'^{14})$-support set $S_k$ for $k=i_l,i_l+1,\ldots,i_r$, where $p_1,p_2,p_3$ are (upper bounded by) $k$-independent and $\epsilon$-dependent polynomials in $1/\delta$, and $c$ is a sufficiently small absolute constant. After the last iteration, we solve the convex program $e_i=\min\tr(\sigma H'_i)$, whose variable $\sigma$ is a density matrix on $\Span S_{i_r}$. This convex program is of polynomial (in $1/\delta$) size because (i) $\Span S_{i_r}$ is of polynomial dimension; (ii) any element in $\Span S_{i_r}$ is an MPS of polynomial bond dimension. Clearly, $e_i\ge\lambda(H_i')$. By definition, there exists a state $|\psi_i\rangle\in\mathcal H$ such that (i) $\langle\psi_i|H^{\le t}|\psi_i\rangle\le\lambda(H^{\le t})+c\epsilon'^7\le\lambda(H^{\le t})+1$; (ii) its reduced density matrix on $\mathcal H_{[i_l,i_r]}$ is supported on $\Span S_{i_r}$. Property (ii) implies $e_i\le\langle\psi_i|H'_i|\psi_i\rangle$ and (i) implies $\langle\psi_i|H'_i|\psi_i\rangle\le\lambda(H'_i)+3$. Indeed, let $|\phi_L\rangle,|\phi_i\rangle,|\phi_R\rangle$ be the ground states of $H^+_L=H_L^{\le t}+\sum_{j=i_l-s}^{i_l-2}H_j,H'_i,H^+_R=\sum_{j=i_r+1}^{i_r+s}H_j+H_R^{\le t}$, respectively. Then,
\begin{eqnarray}
&&\langle\phi_L|H^+_L|\phi_L\rangle+\langle\phi_i|H'_i|\phi_i\rangle+\langle\phi_R|H^+_R|\phi_R\rangle+2\ge\langle\phi_L\phi_i\phi_R|H^{\le t}|\phi_L\phi_i\phi_R\rangle\ge\lambda(H^{\le t})\nonumber\\
&&\ge\langle\psi_i|H^{\le t}|\psi_i\rangle-1\ge\langle\psi_i|H^+_L|\psi_i\rangle+\langle\psi_i|H'_i|\psi_i\rangle+\langle\psi_i|H^+_R|\psi_i\rangle-1\nonumber\\
&&\ge\langle\phi_L|H^+_L|\phi_L\rangle+\langle\psi_i|H'_i|\psi_i\rangle+\langle\phi_R|H^+_R|\phi_R\rangle-1.
\end{eqnarray}

Each iteration consists of four steps: \textbf{extension}, \textbf{cardinality reduction}, \textbf{bond truncation}, and \textbf{error reduction}. Table \ref{t} summarizes the evolution of the parameters $s,b,\Delta_f$ or $\Delta_e$ in each iteration of our algorithm.

The analysis below gives
\begin{equation}
p_1=2^{2^{\tilde O(1/\epsilon)}},~p_2=O(1/\delta)2^{\tilde O(\epsilon^{-1/4}\log^{3/4}(1/\delta))},~p_3=(1/\delta)^{O(1)}.
\end{equation}
The  running time of the algorithm is a polynomial in $p_1,p_2,p_3$ so that Lemma \ref{lem5} follows. The remainder of this section explains each step in detail.
\end{proof}

\begin{table}
\caption{Evolution of the parameters in each iteration. The asterisks mark the parameter that is reduced at every step.}
\begin{tabular}{cccccc}
\hline \hline
& $k$ & $s$ & $b$ & $\Delta_f$ & $\Delta_e$\\
\hline
start & $k-1$ & $p_1p_3$ &$p_2p_3$& n/a &$c\epsilon'^{14}$ \\
\textbf{extension} & $k$ & $ dp_1p_3$ &$p_2p_3$ & n/a &$c\epsilon'^{14}$ \\
\textbf{cardinality reduction} & $k$ & $p_1$* &  $dp_1p_2p_3^2$ & $1/1000$ & $1/1000$ \\
\textbf{bond truncation} & $k$ & $p_1$& $p_2$*& $1/20$ & n/a\\
\textbf{error reduction} & $k$ & $p_1p_3$& $p_2p_3$ & n/a & $c\epsilon'^{14}$*\\ 
\hline \hline
\end{tabular} \label{t}
\end{table}

\subsection{Preliminaries}

\begin{lemma} \label{claim:overlap}
(a) $\langle\psi|H^{\le t}|\psi\rangle\le\lambda(H^{\le t})+\eta\epsilon'$ implies $|\langle\psi|\Psi_0'\rangle|\ge|\langle\psi|\Psi_0'\rangle|^2\ge1-\eta$;\\
(b) $|\langle\psi|\phi_1\rangle|\ge1-\eta_1$ and $|\langle\psi|\phi_2\rangle|\ge1-\eta_2$ imply $|\langle\phi_1|\phi_2\rangle|\ge1-2(\eta_1+\eta_2)$;\\
(c) $|\langle\psi|\phi\rangle|\ge1-\eta$ implies $|\langle\psi|\hat O|\psi\rangle-\langle\phi|\hat O|\phi\rangle|\le2\sqrt{2\eta}$ for any operator $\hat O$ with $\|\hat O\|\le1$.
\end{lemma}

\begin{proof}
Just a few lines of algebra.
\end{proof}

\begin{definition} [truncation]
Let $|\psi\rangle=\sum_{j\ge1}\lambda_j|l_j\rangle|r_j\rangle$ be the Schmidt decomposition of a state $|\psi\rangle\in\mathcal H$ across the cut $i|i+1$, where the Schmidt coefficients are in nonascending order: $\lambda_1\ge\lambda_\ge\cdots>0$. Define $\trunc_D|\psi\rangle=\sum_{j=1}^D\lambda_j|l_j\rangle|r_j\rangle$.
\end{definition}

\begin{lemma} [Eckart-Young theorem] \label{lem:eckart}
The state $|\psi'\rangle=\trunc_D|\psi\rangle/\|\trunc_D|\psi\rangle\|$ satisfies $\langle\psi'|\psi\rangle\ge|\langle\phi|\psi\rangle|$ for any state $|\phi\rangle\in\mathcal H$ of Schmidt rank $D$ (across the cut $i|i+1$).
\end{lemma}

\begin{lemma} [\cite{LVV13}] \label{l:trim2}
Suppose $|\phi\rangle\in\mathcal H$ is a state of Schmidt rank $D$ (across the cut $i|i+1$).
\begin{equation}
|\langle\trunc_{D/\eta}\psi|\phi\rangle|\ge|\langle\psi|\phi\rangle|-\eta,~\forall\eta>0,\psi\in\mathcal H.
\end{equation}
\end{lemma}

Fix a cut $i|i+1$ with $i_l\le i\le i_r-1$ in the definition (\ref{truncated}) of $H^{\le t}$.

\begin{lemma} \label{constantbondapprox}
Set $s\ge\tilde O(\epsilon^{-1})+O(\epsilon^{-1/4}\log^{3/4}(1/\eta))$. Then, $\langle\psi|\Psi_0'\rangle\ge1-\eta$ for $|\psi\rangle=\trunc_{B_\eta}|\Psi_0'\rangle/\|\trunc_{B_{\eta}}|\Psi_0'\rangle\|$, where $B_\eta=2^{\tilde O(1/\epsilon+\epsilon^{-1/4}\log^{3/4}(1/\eta))}$.
\end{lemma}

\begin{proof}
This is a by-product of the proof of the area law for entanglement \cite{Hua14}.
\end{proof}

\begin{lemma} [\cite{DMRG}] \label{l:2}
$\langle\psi|H^{\le t}|\psi\rangle\le\lambda(H^{\le t})+\eta\epsilon'$ and $\eta\le1/10$ imply $\langle\psi'|H^{\le t}|\psi'\rangle\le\lambda(H^{\le t})+25\sqrt{\eta}$ for $|\psi'\rangle=\trunc_{B_\eta}|\psi\rangle/\|\trunc_{B_\eta}|\psi\rangle\|$.
\end{lemma}

\subsection{Algorithm and analysis}

In the $(k-i_l)$th iteration the algorithm constructs a $(k,p_1p_3,p_2p_3,\Delta_e=c\epsilon'^{14})$-support set $S_k$ from a $(k-1,p_1p_3,p_2p_3,\Delta_e=c\epsilon'^{14})$-support set $S_{k-1}$ returned in the previous iteration. \textbf{Extension} is trivial (see \cite{LVV13}), and constructs a $(k,dp_1p_3,p_2p_3,\Delta_e=c\epsilon'^{14})$-support set $S_k^{(1)}$.

\subsubsection{Cardinality reduction}

Let $\tr_{[i,j]}\rho$ denote the partial trace over $\mathcal H_{[i,j]}$ of a density matrix $\rho$ on $\mathcal H$.

\begin{definition} [boundary contraction]
Let $|\psi\rangle=\sum_{j=1}^B\lambda_j^{[i_l]}|l_j^{[i_l]}\rangle|r_j^{[i_l]}\rangle=\sum_{j=1}^B\lambda_j^{[k+1]}|l_j^{[k+1]}\rangle|r_j^{[k+1]}\rangle$ be the Schmidt decompositions of a state $|\psi\rangle\in\mathcal H$ across the cuts $i_l-1|i_l$ and $k|k+1$, respectively. Let $\{|j\rangle\}_{j=1}^B$ be the computational basis of $\mathbf C^B$. Let $U_R(\psi):\mathbf C^B\rightarrow\mathcal{H}_{[k+1,n]}$ be the isometry specified by $U_R(\psi)|j\rangle=|r_j^{[k+1]}\rangle$ such that $U^{-1}_R(\psi)|\psi\rangle=\sum_{j=1}^B\lambda_j^{[k+1]}|l_j^{[k+1]}\rangle|j\rangle\in\mathcal H_{[1,k]}\otimes\mathbf C^B$. The right boundary contraction $C_R(\psi)$ is a density matrix on $\mathcal H_{[k,k]}\otimes\mathbf C^B$:
\begin{equation}
C_R(\psi):=U_R^{-1}(\psi)\tr_{[1,k-1] }(|\psi\rangle\langle\psi|)U_R(\psi).
\end{equation}
Similarly, let $U_L(\psi):\mathbf C^B\rightarrow\mathcal{H}_{[1,i_l-1]}$ be the isometry specified by $U_L(\psi)|j\rangle=|l_j^{[i_l]}\rangle$ such that $U_L^{-1}(\psi)|\psi\rangle=\sum_{j=1}^B\lambda_j^{[i_l]}|j\rangle|r_j^{[i_l]}\rangle\in\mathbf C^B\otimes\mathcal H_{[i_l,n]}$. The left boundary contraction $C_L(\psi)$ is a density matrix on $\mathbf C^B\otimes\mathcal H_{[i_l,i_l]}$:
\begin{equation}
C_L(\psi):=U_L^{-1}(\psi)\tr_{[i_l+1,n] }(|\psi\rangle\langle\psi|)U_L(\psi).
\end{equation}
\end{definition}

Let $H_L:=H_L^{\le t}+\sum_{j=i_l-s}^{i_l-2}H_j,H_M:=\sum_{j=i_l}^{k-1}H_j,H_R:=\sum_{j=k+1}^{i_r+s}H_j+H_R^{\le t}$. Define $H'_L=H_L-\lambda(H_L),H'_M:=H_M-\lambda(H_M),H'_R:=H_R-\lambda(H_R)$ so that $\lambda(H'_L)=\lambda(H'_M)=\lambda(H'_R)=0$.

\begin{lemma} \label{l:gluing}
Let $\rho$ be a density matrix on $\mathbf C^B\otimes\mathcal H_{[i_l,k]}\otimes\mathbf C^B$, and $|\psi\rangle=\sum_{j=1}^B\lambda_j^{[i_l]}|l_j^{[i_l]}\rangle|r_j^{[i_l]}\rangle=\sum_{j=1}^B\lambda_j^{[k+1]}|l_j^{[k+1]}\rangle|r_j^{[k+1]}\rangle$ be the Schmidt decompositions of a state $|\psi\rangle\in\mathcal H$ across the cuts $i_l-1|i_l$ and $k|k+1$, respectively. The density matrix $\rho':=U_L(\psi)U_R(\psi)\rho U_R^{-1}(\psi)U_L^{-1}(\psi)$ on $\mathcal H$ has energy
\begin{eqnarray}
&&\tr(\rho' H^{\le t})\le\tr(\rho H_M)+\langle\psi|(H_L+H_{i_l-1}+H_k+H_R)|\psi\rangle\nonumber\\
&&+\|\tr_{[i_l+1,k]\otimes B}\rho-C_L(|\psi\rangle)\|_1\left(1+\max_{|l\rangle\in\Span\left\{\left|l_j^{[i_l]}\right\rangle\right\}}\|H'_L|l\rangle\|\right)\nonumber\\
&&+\|\tr_{B\otimes[i_l,k-1]}\rho-C_R(|\psi\rangle)\|_1\left(1+\max_{|r\rangle\in\Span\left\{\left|r_j^{[k+1]}\right\rangle\right\}}\|H'_R|r\rangle\|\right),
\end{eqnarray}
where $\tr_{[i_l+1,k]\otimes B}\rho$ denotes the partial trace over $\mathcal H_{[i_l+1,k]}\otimes\mathbf C^B$, and similarly for $\tr_{B\otimes[i_l,k-1]}\rho$.
\end{lemma}

\begin{proof}
Write $\tr(\rho'H^{\le t})=\tr(\rho'H_M)+\tr(\rho'(H_L+H_{i_l-1}))+\tr(\rho'(H_k+H_R))$. For the first term, $\tr(\rho' H_M)=\tr(\rho H_M)$ as $U_L(\psi),U_R(\psi)$ are isometries. For the third term,
\begin{eqnarray}
&&\tr(\rho'(H_k+H_R))-\langle\psi|(H_k+H_R)|\psi\rangle=\tr[(\rho'-|\psi\rangle\langle\psi|)(H_k+H_R)]\nonumber\\
&&=\tr[(\rho'-|\psi\rangle\langle\psi|)(H_k+H'_R)]=\tr[\tr_{[1,k-1]}(\rho'-|\psi\rangle\langle\psi|)(H_k+H'_R)]\nonumber\\
&&=\tr[U_R^{-1}(\psi)\tr_{[1,k-1]}(\rho'-|\psi\rangle\langle\psi|)U_R(\psi)U_R^{-1}(\psi)(H_k+H'_R)U_R(\psi)]\nonumber\\
&&=\tr[(\tr_{B\otimes[i_l,k-1]}\rho-C_R(\psi))U_R^{-1}(\psi)(H_k+H'_R)U_R(\psi)]\nonumber\\
&&\le\|\tr_{B\otimes[i_l,k-1]}\rho-C_R(\psi)\|_1\|U_R^{-1}(\psi)H_kU_R(\psi)+U_R^{-1}(\psi)H'_RU_R(\psi)\|\nonumber\\
&&\le\|\tr_{B\otimes[i_l,k-1]}\rho-C_R(\psi)\|_1\left(1+\max_{|r\rangle\in\Span\left\{\left|R_j^{[k+1]}\right\rangle\right\}}\|H'_R|r\rangle\|\right).
\end{eqnarray}
Similarly, for the second term,
\begin{equation}
\tr(\rho'(H_L+H_{i_l-1}))-\langle\psi|(H_L+H_{i_l-1})|\psi\rangle\le\|\tr_{[i_l+1,k]\otimes B}\rho-C_L(\psi)\|_1\left(1+\max_{|l\rangle\in\Span\left\{\left|l_j^{[i_l]}\right\rangle\right\}}\|H'_L|l\rangle\|\right).
\end{equation}
\end{proof}

Let $N_l$ and $N_r$ be, respectively, $\xi$-nets with $\xi=\tilde\Omega(\epsilon)$ for the trace norm over the space of left and right boundary contractions of bond dimension $B_{10\sqrt c\epsilon'^6}=2^{\tilde O(1/\epsilon)}$ such that $|N_l|=|N_r|=(B/\xi)^{O(B)}=2^{2^{\tilde O(1/\epsilon)}}$. It is straightforward to construct $N_l$ and $N_r$ in time $\poly|N_l|=\poly|N_r|=2^{2^{\tilde O(1/\epsilon)}}$. 

\noindent
=======================================================\\
\textbf{Cardinality reduction} convex program and \textbf{bond truncation}\\
--------------------------------------------------------------------------------------------------------------------------------\\ 
0. Let the variable $\rho$ be a density matrix on $\mathbf C^{B_{10\sqrt c\epsilon'^6}}\otimes\Span S_i^{(1)}\otimes\mathbf C^{B_{10\sqrt c\epsilon'^6}}\subseteq\mathbf C^{B_{10\sqrt c\epsilon'^6}}\otimes\mathcal H_{[i_l,k]}\otimes\mathbf C^{B_{10\sqrt c\epsilon'^6}}$.\\
1. For each $X_l\in N_l$ and each $X_r\in N_r$, solve the convex program:
\begin{equation} \label{covpro}
\min~\tr(\rho H_M);~\mathrm{s.~t.}~\|\tr_{[i_l+1,k]\otimes B}\rho-X_l\|_1\le\xi,~\|\tr_{B\otimes[i_l,k-1]}\rho-X_r\|_1\le\xi,~\tr\rho=1,~\rho\ge0.
\end{equation} 
2. Let $|\varphi\rangle=\sum_{j_l,j_r}|j_l\rangle|\varphi_{j_l,j_r}\rangle|j_r\rangle$ be the eigenvector of the solution $\rho$ with the largest eigenvalue.\\
3. Let $|\varphi'\rangle=\sum_{j_l,j_r}|j_l\rangle|\varphi'_{j_l,j_r}\rangle|j_r\rangle$ be the state obtained by truncating each bond (in whatever order) of $|\varphi\rangle$ to $p_2$.\\
4. $S_k^{(3)}$ consists of the MPS representations of all $|\varphi'_{j_l,j_r}\rangle$.\\
=======================================================

Let $P_{t'}$ be the projection onto the subspace $(\mathcal H_{[1,i_l-1]}\otimes\mathcal H_{[k+1,n]})^{\le t'}$ spanned by the eigenvectors of $H'_L+H'_R$ with eigenvalues at most $t'$, and let $Q_{t'}$ be the projection onto the subspace spanned by the eigenvectors of $H'_L+H'_R+H'_M$ with eigenvalues at most $t'$.

\begin{lemma} [truncation lemma]
\begin{equation}
\|(1-P_{t'})|\Psi_0'\rangle\|\le\|(1-Q_{t'})|\Psi_0'\rangle\|\le100\cdot2^{-t'/20}.
\end{equation}
\end{lemma}

\begin{proof}
The first inequality is obvious: $P_{t'}\ge Q_{t'}$ as $[H'_L+H'_R,H'_M]=0$ and $H'_M\ge0$. The second inequality was proved in \cite{AKLV13}.
\end{proof}

Let $t'=O(\log(1/\epsilon))$ such that $100\cdot2^{-t'/20}\le c\epsilon'^{14}$.

\begin{lemma}
There exists a state $|\psi\rangle\in\mathcal H_{[1,i_l-1]}^{\le t'}\otimes\Span S_k^{(1)}\otimes\mathcal H_{[k+1,n]}^{\le t'}$ such that (i) its Schmidt ranks across the cuts $i_l-1|i_l$ and $k|k+1$ are upper bounded by $B_{10\sqrt c\epsilon'^6}$; (ii) $\langle\psi|H^{\le t}|\psi\rangle\le\lambda(H^{\le t})+250c^{1/8}\epsilon'$.
\end{lemma}

\begin{proof}
Let $|\phi\rangle$ be a witness for $S_k^{(1)}$. Since $S_k^{(1)}$ is a $(k,dp_1p_3,p_2p_3,\Delta_e=c\epsilon'^{14})$-support set. Lemma \ref{claim:overlap}(a) implies $|\langle\phi|\Psi_0'\rangle|\ge1-c\epsilon'^{14}$. Lemma \ref{claim:overlap}(c) implies
\begin{eqnarray}
&&\langle\phi|(H_{i_l-1}+H_k)|\phi\rangle\ge\langle\Psi_0'|(H_{i_l-1}+H_k)|\Psi_0'\rangle-4\sqrt{2c}\epsilon'^7\nonumber\\
&&\Rightarrow\langle\phi|(H'_L+H'_M+H'_R)|\phi\rangle\le\langle\Psi_0'|(H'_L+H'_M+H'_R)|\Psi_0'\rangle+4\sqrt{2c}\epsilon'^7+c\epsilon'^{14}.
\end{eqnarray}
Let $|\phi'\rangle=P_{t'}|\phi\rangle/\|P_{t'}|\phi\rangle\|$ so that $|\phi'\rangle\in(\mathcal H_{[1,i_l-1]}\otimes\mathcal H_{[k+1,n]})^{\le t'}\otimes\Span S_k^{(1)}\subseteq\mathcal H_{[1,i_l-1]}^{\le t'}\otimes\Span S_k^{(1)}\otimes\mathcal H_{[k+1,n]}^{\le t'}$ by construction. 
\begin{eqnarray}
&&|\langle\phi'|\Psi_0'\rangle|\ge|\langle\phi|P_{t'}|\Psi_0'\rangle|\ge|\langle\phi|\Psi_0'\rangle|-|\langle\phi|(1-P_{t'})|\Psi_0'\rangle|\ge1-c\epsilon'^{14}-\|(1-P_{t'})|\Psi_0'\rangle\|\nonumber\\
&&\ge1-c\epsilon'^{14}-100\cdot2^{-t'/20}\ge1-2c\epsilon'^{14}\nonumber\\
&&\Rightarrow\langle\phi'|(H_{i_l-1}+H_k)|\phi'\rangle\le\langle\Psi_0'|(H_{i_l-1}+H_k)|\Psi_0'\rangle+4\sqrt{c}\epsilon'^7\nonumber\\
&&\langle\Psi_0'|(H'_L+H'_M+H'_R)|\Psi_0'\rangle+6\sqrt c\epsilon'^7\ge\langle\phi|(H'_L+H'_M+H'_R)|\phi\rangle\nonumber\\
&&=\langle\phi|P_{t'}(H'_L+H'_M+H'_R)P_{t'}|\phi\rangle+\langle\phi|(1-P_{t'})(H'_L+H'_M+H'_R)(1-P_{t'})|\phi\rangle\nonumber\\
&&\ge\langle\phi'|(H'_L+H'_M+H'_R)|\phi'\rangle\|P_{t'}|\phi\rangle\|^2+t'\|(1-P_{t'})|\phi\rangle\|^2
\end{eqnarray}
Clearly, $\lambda(H'_L+H_{i_l-1}+H'_M+H_k+H'_R)\le\|H_{i_l-1}+H_k\|\le2$. Hence, $\langle\Psi_0'|(H'_L+H'_M+H'_R)|\Psi_0'\rangle\le2$ as $H_{i_l-1},H_k\ge0$, and
\begin{eqnarray}
&&t'\gg\langle\Psi_0'|(H'_L+H'_R+H'_M)|\Psi_0'\rangle+6\sqrt c\epsilon'^7\nonumber\\
&&\Rightarrow\langle\phi'|(H'_L+H'_M+H'_R)|\phi'\rangle\le\langle\Psi_0'|(H'_L+H'_M+H'_R)|\Psi_0'\rangle+6\sqrt c\epsilon'^7\nonumber\\
&&\Rightarrow\langle\phi'|H^{\le t}|\phi'\rangle\le\lambda(H^{\le t})+6\sqrt c\epsilon'^7+4\sqrt c\epsilon'^7\le\lambda(H^{\le t})+10\sqrt c\epsilon'^7.
\end{eqnarray}
Lemma \ref{l:2} implies that the state $|\phi''\rangle:=\trunc_{B_{10\sqrt c\epsilon'^6}}|\phi'\rangle/\|\trunc_{B_{10\sqrt c\epsilon'^6}}|\phi'\rangle\|\in\mathcal H^{\le t'}_{[1,i_l-1]}\otimes\Span S_k^{(1)}\otimes\mathcal H^{\le t'}_{[k+1,n]}$ (the truncation is across the cut $i_l-1|i_l$) has energy $\langle\phi''|H^{\le t}|\phi''\rangle\le\lambda(H^{\le t})+100c^{1/4}\epsilon'^3$. Once again, the state $|\psi\rangle:=\trunc_{B_{100c^{1/4}\epsilon'^2}}|\phi''\rangle/\|\trunc_{B_{100c^{1/4}\epsilon'^2}}|\phi''\rangle\|\in\mathcal H^{\le t'}_{[1,i_l-1]}\otimes\Span S_k^{(1)}\otimes\mathcal H^{\le t'}_{[k+1,n]}$ (the truncation is across the cut $k|k+1$) has energy $\langle\psi|H^{\le t}|\psi\rangle\le\lambda(H^{\le t})+250c^{1/8}\epsilon'$.
\end{proof}

\begin{lemma}
$S_k^{(2)}$ is a $(k,p_1,dp_1p_2p_3^2,\Delta_e=1/1000)$-support set, where $S_k^{(2)}$ consists of the MPS representations of all $|\varphi_{j_l,j_r}\rangle$.
\end{lemma}

\begin{proof}
Since $N_l$ and $N_r$ are $\xi$-nets, there are elements $X_l\in N_l$ and $X_r\in N_r$ such that $\|-C_L(\psi)-X_l\|_1\le\xi$ and $\|C_R(\psi)-X_r\|_1\le\xi$, respectively. Clearly, $\tr(\rho H_M)\le\langle\psi|H_M|\psi\rangle$ as $U_L(\psi)U_R(\psi)|\psi\rangle\langle\psi|U_R^{-1}(\psi)U_L^{-1}(\psi)$ is a feasible solution to the convex program (\ref{covpro}). Let $\sigma=U_L(\psi)U_R(\psi)\rho U_R^{-1}(\psi)U_L^{-1}(\psi)$, and set $t=\tilde\Omega(\epsilon)$ such that $4\xi(1+t')\le\epsilon'/4000$. Lemma \ref{l:gluing} implies
\begin{eqnarray} \label{eq:leading-1}
&&\tr(\sigma H^{\le t})\le\tr(\rho H_M)+\langle\psi|(H_L+H_{i_l-1}+H_k+H_R)|\psi\rangle\nonumber\\
&&+\|\tr_{[i_l+1,k]\otimes B}\rho-C_L(\psi)\|_1\left(1+\max_{|l\rangle\in\Span\left\{\left|l_j^{[i_l]}\right\rangle\right\}}\|H'_L|l\rangle\|\right)\nonumber\\
&&+\|\tr_{B\otimes[i_l,k-1]}\rho-C_R(\psi)\|_1\left(1+\max_{|r\rangle\in\Span\left\{\left|r_j^{[k+1]}\right\rangle\right\}}\|H'_R|r\rangle\|\right)\nonumber\\
&&\le\langle\psi|(H_L+H_{i_l-1}+H_M+H_k+H_R)|\psi\rangle+2\xi\left(2+\max_{|l\rangle\in\mathcal H_{[1,i_l-1]}^{\le t'}}\|H'_L|l\rangle\|+\max_{|r\rangle\in\mathcal H_{[k+1,n]}^{\le t'}}\|H'_R|r\rangle\|\right)\nonumber\\
&&\le\langle\psi|H^{\le t}|\psi\rangle+4\xi(1+t')\le\lambda(H^{\le t})+250c^{1/8}\epsilon'+\epsilon'/4000\le\lambda(H^{\le t})+\epsilon'/2000
\end{eqnarray}
for sufficiently small constant $c$. We observe that\\
(1) there exists at least an eigenstate of $\sigma$ with energy (with respect to $H^{\le t}$) at most $\lambda(H^{\le t})+\epsilon'/1000$;\\
(2) there is at most one such eigenstate as Lemma \ref{claim:overlap}(a) implies that such an eigenstate is close to $|\Psi_0'\rangle$;\\
(3) this eigenstate (denoted by $|\Phi\rangle$) has the largest eigenvalue due to Markov's inequality in probability theory;\\
(4) $|\Phi\rangle=U_L(\psi)U_R(\psi)|\varphi\rangle\in\mathcal H_{[1,i_l-1]}^{\le t'}\otimes\Span S_k^{(1)}\otimes\mathcal H_{[k+1,n]}^{\le t'}$ is a witness for $S_k^{(2)}$ as a $(k,p_1,dp_1p_2p_3^2,\Delta_e=1/1000)$-support set with $p_1=B_{10\sqrt c\epsilon'^6}^2|N|=2^{2^{\tilde O(1/\epsilon)}}$.
\end{proof}

\subsubsection{Bond truncation}

\begin{lemma}
There exists a state $|\psi\rangle$ such that (i) $|\langle\psi|\Psi_0'\rangle|>1- 1/1000$; (ii) the Schmidt rank of $|\psi\rangle$ across the cut $j|j+1$ is upper bounded by $r=2^{\tilde O(\epsilon^{-1/4}\log^{3/4}(1/\delta))}$ for any $i_l\le j\le k-1$.
\end{lemma}

\begin{proof}
We express $|\Psi_0'\rangle$ exactly as an MPS (\ref{MPS}) in the canonical form (\ref{canonical}) with bond dimension $\exp(O(n))$. Let $P=\diag\{1,1,\ldots,1,0,0,\ldots\}$ with $\tr P=r$. We define the unnormalized state $|\psi'\rangle$ as
\begin{equation}
|\psi'\rangle=\sum_{j_1,j_2,\ldots,j_n=1}^dA_{j_1}^{[1]}A_{j_2}^{[2]}\cdots A_{j_{i_l-1}}^{[i_l-1]}A_{j_{i_l}}^{[i_l]}PA_{j_{i_l+1}}^{[i_l+1]}P\cdots PA_{j_{k-1}}^{[k-1]}PA_{j_{k}}^{[k]}A_{j_{k+1}}^{[k+1]}\cdots A_{j_n}^{[n]}|j_1j_2\cdots j_n\rangle.
\end{equation}
Clearly, the Schmidt rank of $|\psi'\rangle$ across the cut $j|j+1$ is upper bounded by $r$ for any $i_l\le j\le k-1$. Let $E^{[i]}(X):=\sum_{j=1}^dA_j^{[i]\dag}XA_j^{[i]}$ be a completely positive trace-preserving (CPTP) linear map. Then,
\begin{equation}
\langle\psi'|\Psi_0'\rangle=\mathrm{tr}X_{k+1},~X_{j+1}:=E^{[j]}(X_jP),~j=i_l+1,i_l+2,\ldots,k,~X_{i_l+1}:=\Lambda^{[i_l+1]}.
\end{equation}
As a CPTP linear map is nonexpansive with respect to the trace norm,
\begin{eqnarray}
&&\|E^{[i]}(X)\|_1\le\|X\|_1\Rightarrow\|\Lambda^{[j+1]}-X_{j+1}\|_1\le\|\Lambda^{[j]}-X_{j}P\|_1\le\|\Lambda^{[j]}(1-P)\|_1+\|\Lambda^{[j]}-X_{j}\|_1\nonumber\\
&&\Rightarrow1-\langle\psi'|\Psi_0'\rangle=\tr(\Lambda^{[k+1]}-X_{k+1})\le\|\Lambda^{[k+1]}-X_{k+1}\|_1\le\varepsilon_r(k-i_l)\le\varepsilon_r/\delta\le1/1000,
\end{eqnarray}
where $\varepsilon_r=\delta/1000$ for $r=2^{\tilde O(\epsilon^{-1/4}\log^{3/4}\delta^{-1})}$. Similarly,
\begin{eqnarray}
&&\langle\psi'|\psi'\rangle=\mathrm{tr}X'_{k+1},~X'_{j+1}:=E^{[j]}(PX'_jP),~j=i_l+1,i_l+2,\ldots,k,~X'_{i_l+1}:=\Lambda^{[i_l+1]}\nonumber\\
&&\Rightarrow\langle\psi'|\psi'\rangle\le\|X'_{k+1}\|_1\le\cdots\le\|X'_{i_l+1}\|_1=1.
\end{eqnarray}
Hence $|\psi\rangle:=|\psi'\rangle/\||\psi\rangle\|$ satisfies $\langle\psi|\Psi_0'\rangle\ge999/1000$.
\end{proof}

\begin{lemma} 
$S_k^{(3)}$ is a $(k,p_1,p_2,\Delta_f=1/20)$-support set.
\end{lemma}

\begin{proof} 
Since $|\Phi\rangle$ is a witness for $S_k^{(2)}$ with energy $|\langle\Phi|H^{\le t}|\Phi\rangle|\le\lambda(H^{\le t})+\epsilon'/1000$. Lemma \ref{claim:overlap}(a) implies $|\langle\Phi|\Psi_0'\rangle|\ge999/1000$. Lemma \ref{claim:overlap}(b) implies $|\langle\Phi|\psi\rangle|\ge249/250$. Let $\Phi'$ be the state obtained by truncating the each bond $j|j+1$ for $j=i_l,\dots,k-1$ of $\Psi$ to $p_2=1000r/\delta$. Then Lemma \ref{l:trim2} implies $|\langle\psi|\Phi'\rangle|\ge199/200$. Lemma \ref{claim:overlap}(b) implies $|\langle\Psi_0'|\Phi'\rangle|\ge49/50$. Hence $|\Phi'\rangle=U_L(\psi)U_R(\psi)|\varphi'\rangle$ is a witness for $S_k^{(3)}$ as a $(k,p_1,p_2,\Delta_f=1/20)$-support set.
\end{proof}

\subsubsection{Error reduction}

Let $\eta=c\epsilon'^{14}$. Assume for the moment that we have an estimate $\epsilon'_0$ of $\lambda(H^{\le t})$ in the sense that $\xi:=|\epsilon'_0-\lambda(H^{\le t})|\le\epsilon'/\sqrt q\le\epsilon'/2$, where $q=4\log(1/\eta)+24$. Let
\begin{equation}
A:=\exp\left(-\frac{q(H^{\le t}-\epsilon'_0)^2}{2\epsilon'^2}\right)=\frac{\epsilon'}{\sqrt{2\pi q}}\int_{-\infty}^{+\infty}\exp\left(-\frac{\epsilon'^2\tau^2}{2q}-i(H^{\le t}-\epsilon'_0)\tau\right)\mathrm{d}\tau.
\end{equation}
Recall that $|\Phi'\rangle$ is a witness for $S_k^{(3)}$ such that $|\langle\Phi_0'|\Phi'\rangle|\ge19/20$.

\begin{lemma} [\cite{DMRG}] \label{l1}
$\langle\phi|H^{\le t}|\phi\rangle\le\lambda(H^{\le t})+\eta\epsilon'/100$ for $|\phi\rangle=A|\Phi'\rangle/\|A|\Phi'\rangle\|$.
\end{lemma}

Let $K$ be an MPO of bond dimension $D$ across the cuts $j|j+1$ for $j=i_l-1,i_l,\ldots,k$ (its bond dimension across other cuts may be very large). It is straightforward to decompose $K$ as a sum of $D^2$ terms:
\begin{equation} \label{dec}
K=\sum_{j_l,j_r=1}^DK_{j_l}^L\otimes K_{j_l,j_r}^M\otimes K_{j_r}^R,
\end{equation}
where $K_{j_l}^L,K_{j_r}^R$ are MPO on $\mathcal H_{[1,i_l-1]},\mathcal H_{[k+1,n]}$, respectively, and $K_{j_l,j_r}^M$ is an MPO of bond dimension $D$ on $\mathcal H_{[i_l,k]}$.

\begin{lemma} \label{lem18}
There is an MPO $K_T$ such that (i) $\|K_T-\exp(-iH^{\le t}T)\|\le\xi'$; (ii) the bond dimension across the cuts $j|j+1$ for $i_l-1\le j\le k$ is upper bounded by $2^{O(T)}/\poly(\delta\xi')$; (iii) $K_{j_l,j_r}^M$'s in the decomposition (\ref{dec}) for $K_T$ can be constructed in time $2^{O(T)}/\poly(\delta\xi')$.
\end{lemma}

\begin{proof}
We give a constructive proof so that (iii) will be apparent. Recall that $H^{\le t}$ is local on $\mathcal{H}_{[i_l-s+1,i_r+s]}$ and has bounded norm on the left and the right of this region. Following \cite{Osb06}, for each $i_l-s/2\le j\le i_r+s/2$ consider the unitary operator
\begin{equation}
V(T):=\exp\left(-iH_L^{\le T}T-i\sum_{l=i_l-s}^{j-1}H_lT\right)\exp\left(-i\sum_{l=j+1}^{i_r+s}H_lT-iH_R^{\le T}T\right)\exp(iH^{\le t}T).
\end{equation}
It is straightforward to verify that $V(T)=\mathcal{T}\exp\left(-i\int_0^TL(\tau)\mathrm{d}\tau\right)$ is the time-evolution operator of the ``Hamiltonian'' $L(\tau)=\exp(-iH^{\le t}\tau)H_j\exp(iH^{\le t}\tau)$, where $\mathcal T$ is the time-ordering operator. Define $V'(T)=\mathcal{T}\exp\left(-i\int_0^TL'(\tau)\mathrm{d}\tau\right)$ with
\begin{equation}
L'(\tau)=\exp\left(-i\sum_{l=j-s'+1}^{j+s'-1}H_l\tau\right)H_j\exp\left(i\sum_{l=j-s'+1}^{j+s'-1}H_l\tau\right),~s'\le s/4.
\end{equation}
Since the ``Hamiltonian'' $L'(\tau)$ acts on $\mathcal H_{[j-s'+1,j+s']}$, $V'(T)$ also acts on $\mathcal H_{[j-s'+1,j+s']}$. Moreover, as a consequence of the Lieb-Robinson bound (see Ref. 24 in \cite{Osb06} for a simple direct proof)
\begin{equation}
\|L(T)-L'(T)\|=\exp(-\Omega(s'))\Rightarrow\|V(T)-V'(T)\|=T\exp(-\Omega(s'))=\exp(-\Omega(s'))
\end{equation}
for $s'>O(T)$. Hence,
\begin{equation}
\exp(iH^{\le t}T)=\exp\left(iH_L^{\le T}T+i\sum_{l=i_l-s}^{j-1}H_lT\right)\exp\left(i\sum_{l=j+1}^{i_r+s}H_lT+iH_R^{\le T}T\right)V'(T)+\exp(-\Omega(s')).
\end{equation}
Using this decomposition sequentially for $j=i_l-1+2ms',~m=0,1,2,\ldots,m_0$ with $m_0=\left[\frac{i_r-i_l}{2s'}\right]+1$, we obtain
\begin{eqnarray}
&&\exp(iH^{\le t}T)=\prod_{m=0}^{m_0+1}U_m\prod_{m=0}^{m_0}V_m+m_0\exp(-\Omega(s')),\nonumber\\
&&U_{1\le m\le m_0}=\exp\left(i\sum_{l=i_l+2(m-1)s'}^{i_l-2+2ms'}H_lT\right),\nonumber\\
&&U_0=\exp\left(iH_L^{\le T}T+i\sum_{l=i_l-s}^{i_l-2}H_lT\right),~U_{m_0+1}=\exp\left(i\sum_{l=i_l+2m_0s'}^{i_r+s}H_lT+iH_R^{\le T}T\right),
\end{eqnarray}
and $V_m$ is an operator acting on $\mathcal{H}_{[i_l+(2m-1)s',i_l+(2m+1)s'-1]}$. Let $K_T=\prod_{m=0}^{m_0+1}U_m\prod_{m=0}^{m_0}V_m$ with $s'=O(T+\log(m_0/\xi'))$ such that (i) holds. Clearly, the bond dimension of $K_T$ across the cut $j|j+1$ for each $i_l-1\le j\le k$ is upper bounded by $2^{O(s')}=2^{O(T)}/\poly(\delta\xi')$.
\end{proof}

\begin{lemma} \label{l2}
There is an MPO $K$ such that (i) $\|K-A\|\le\eta\epsilon\delta/1000=:\eta'$; (ii) the bond dimension across the cuts $j|j+1$ for $i_l-1\le j\le k$ is upper bounded by $D=(1/\delta)^{O(1+\sqrt{\log(1/\eta)/\log(1/\delta)}/\epsilon)}$; (iii) $K_{j_l,j_r}^M$'s in the decomposition (\ref{dec}) for $K$ can be constructed in time $\poly(D)$. 
\end{lemma}

\begin{proof}
Following \cite{Has09}, we truncate and discretize the integral in the definition of $A$. Let $t_j=\tau j$ with $j=0,\pm1,\ldots,\pm T/\tau$, and replace the integral by a sum over $t_j$. The truncation error is of order $\exp\left(-\frac{\epsilon'^2T^2}{2q}\right)\le\eta'/3$ for $T=O(\sqrt{q\log(1/\eta')}/\epsilon')$. The discretization error is of order $\epsilon'\tau\|H^{\le t}\|T/\sqrt{2\pi q}\le\eta'/3$ for $\tau=O\left(\frac{\delta\eta'}{\sqrt{\log(1/\eta')}}\right)$. For each $t_j$, Lemma \ref{lem18} implies that the propagator $\exp(-i(H^{\le t}-\epsilon'_0)t_j)$ can be approximated to error $\eta'/3$ by an MPO of bond dimension $2^{O(t_j)}\poly(\delta/\eta')=2^{O(T)}\poly(\delta/\eta')=2^{O(\sqrt {q\log(1/\eta')}/\epsilon')}\poly(\delta/\eta')$ across the cuts $j|j+1$ for $i_l-1\le j\le k$. The number of terms is $2T/\tau+1=O\left(\frac{\sqrt q\log(1/\eta')}{\delta\eta'\epsilon'}\right)$. Hence the bond dimension across the cut $j|j+1$ of the MPO $K$ for $i_l-1\le j\le k$ is upper bounded by $D=(1/\delta)^{O(1+\sqrt{\log(1/\eta)/\log(1/\delta)}/\epsilon)}$. The efficiency of constructing $K_{j_l,j_r}^M$'s in the decomposition (\ref{dec}) for $K$ follows from Lemma \ref{lem18}(iii).
\end{proof}

Lemmas \ref{l1}, \ref{l2} imply $\langle\phi'|H^{\le t}|\phi'\rangle\le\lambda(H^{\le t})+\eta\epsilon'$ for $|\phi'\rangle=K|\Phi'\rangle/\|K|\Phi'\rangle\|$. It is easy to see that $S_k:=\{K_{j_l,j_r}^M|\psi\rangle:\forall j_l,j_r,|\psi\rangle\in S_k^{(3)}\}$ is a $(k,D^2p_1,Dp_2,\Delta_e=\eta)$-support set. Hence, $p_3=(1/\delta)^{O(1+\sqrt{\log(1/\eta)/\log(1/\delta)}/\epsilon)}=(1/\delta)^{O(1)}$.

We briefly comment on the assumption that we have an estimate $\epsilon'_0$ of $\lambda(H^{\le t})$ to error $\xi$. Since $0\le\lambda(H^{\le t})\le1/\delta+2s+2t+1=O(1/\delta)$, we run the whole algorithm (for $H^{\le t}$) with $\epsilon'_0=j\xi$ and obtain a ``candidate $e_i$'' (cf. Lemma \ref{lem5}) for each $j=0,1,\ldots,O(\delta^{-1}/\xi)$. The minimum of these candidates is identified as the ``final $e_i$.'' This completes our algorithm.

\section*{Acknowledgment}

We would like to thank Xie Chen and Aram W. Harrow for discussions. In particular, we learned Lemma \ref{lem0} from A.W.H some time ago. This work was supported by DARPA OLE.


\begin{thebibliography}{10}

\bibitem{AAI10}
D.~Aharonov, I.~Arad, and S.~Irani.
\newblock Efficient algorithm for approximating one-dimensional ground states.
\newblock {\em Physical Review A}, 82(1):012315, 2010.

\bibitem{AAV13}
D.~Aharonov, I.~Arad, and T.~Vidick.
\newblock Guest column: The quantum {PCP} conjecture.
\newblock {\em SIGACT News}, 44(2):47--79, 2013.

\bibitem{AGIK07}
D.~Aharonov, D.~Gottesman, S.~Irani, and J.~Kempe.
\newblock The power of quantum systems on a line.
\newblock In {\em Proceedings of the 48th Annual IEEE Symposium on Foundations
  of Computer Science}, pages 373--383, 2007.

\bibitem{AGIK09}
D.~Aharonov, D.~Gottesman, S.~Irani, and J.~Kempe.
\newblock The power of quantum systems on a line.
\newblock {\em Communications in Mathematical Physics}, 287(1):41--65, 2009.

\bibitem{AKLV13}
I.~Arad, A.~Kitaev, Z.~Landau, and U.~Vazirani.
\newblock An area law and sub-exponential algorithm for 1{D} systems.
\newblock arXiv:1301.1162v1.

\bibitem{BH13}
F.~G. S.~L. Brandao and A.~W. Harrow.
\newblock Product-state approximations to quantum ground states.
\newblock In {\em Proceedings of the 45th Annual ACM Symposium on Theory of
  Computing}, pages 871--880, 2013.

\bibitem{CF15}
C.~T. Chubb and S.~T. Flammia.
\newblock Computing the degenerate ground space of gapped spin chains in
  polynomial time.
\newblock arXiv:1502.06967.

\bibitem{FNW92}
M.~Fannes, B.~Nachtergaele, and R.~Werner.
\newblock Finitely correlated states on quantum spin chains.
\newblock {\em Communications in Mathematical Physics}, 144(3):443--490, 1992.

\bibitem{Has09}
M.~B. Hastings.
\newblock Quantum adiabatic computation with a constant gap is not useful in
  one dimension.
\newblock {\em Physical Review Letters}, 103(5):050502, 2009.

\bibitem{Hua14}
Y.~Huang.
\newblock Area law in one dimension: Degenerate ground states and {R}enyi
  entanglement entropy.
\newblock arXiv:1403.0327.

\bibitem{DMRG}
Y.~Huang.
\newblock A polynomial-time algorithm for the ground state of one-dimensional
  gapped {H}amiltonians.
\newblock arXiv:1406.6355.

\bibitem{LVV13}
Z.~Landau, U.~Vazirani, and T.~Vidick.
\newblock A polynomial-time algorithm for the ground state of 1{D} gapped local
  {H}amiltonians.
\newblock arXiv:1307.5143v1.

\bibitem{Osb06}
T.~J. Osborne.
\newblock Efficient approximation of the dynamics of one-dimensional quantum
  spin systems.
\newblock {\em Physical Review Letters}, 97(15):157202, 2006.

\bibitem{Osb12}
T.~J. Osborne.
\newblock Hamiltonian complexity.
\newblock {\em Reports on Progress in Physics}, 75(2):022001, 2012.

\bibitem{PVWC07}
D.~Perez-Garcia, F.~Verstraete, M.~M. Wolf, and J.~I. Cirac.
\newblock Matrix product state representations.
\newblock {\em Quantum Information and Computation}, 7(5):401--430, 2007.

\bibitem{SC10}
N.~Schuch and J.~I. Cirac.
\newblock Matrix product state and mean-field solutions for one-dimensional
  systems can be found efficiently.
\newblock {\em Physical Review A}, 82(1):012314, 2010.

\bibitem{SWVC08}
N.~Schuch, M.~M. Wolf, F.~Verstraete, and J.~I. Cirac.
\newblock Entropy scaling and simulability by matrix product states.
\newblock {\em Physical Review Letters}, 100(3):030504, 2008.

\bibitem{VC06}
F.~Verstraete and J.~I. Cirac.
\newblock Matrix product states represent ground states faithfully.
\newblock {\em Physical Review B}, 73(9):094423, 2006.

\bibitem{Whi92}
S.~R. White.
\newblock Density matrix formulation for quantum renormalization groups.
\newblock {\em Physical Review Letters}, 69(19):2863--2866, 1992.

\bibitem{Whi93}
S.~R. White.
\newblock Density-matrix algorithms for quantum renormalization groups.
\newblock {\em Physical Review B}, 48(14):10345--10356, 1993.

\end{thebibliography}
\end{document}